\documentclass[12pt]{article}
\usepackage{epsfig}
\usepackage{amsfonts}
\usepackage{mathtools} 
\usepackage{amscd}
\usepackage{latexsym}
\usepackage{amsmath,amssymb,amsthm}
\usepackage{verbatim}
\usepackage{setspace}
\usepackage{color}
\usepackage{cite}
\usepackage{graphicx}
\usepackage{mathtools}
\usepackage[all]{xy}
\usepackage{tikz}

\usepackage[textheight=9in, textwidth=6.5in, letterpaper]{geometry}

\usepackage{color}   
\usepackage{hyperref}
\hypersetup{
    colorlinks=true,  
    linktoc=all,     
    linkcolor=black,  
    citecolor=black,
    filecolor=black,
    urlcolor=black,
}

\numberwithin{equation}{section}

\newtheorem{Theorem}{Theorem}[section]
\newtheorem*{Theorem*}{Theorem}

\newtheorem{Lemma}[Theorem]{Lemma}
\newtheorem{Proposition}[Theorem]{Proposition}
 { \theoremstyle{definition}
\newtheorem{Definition}[Theorem]{Definition}

\newtheorem{Remark}[Theorem]{Remark} }

\def\p{\partial}
\def\cl{{\cal L}}

\def\<{\langle}
\def\>{\rangle}

\def\cO{\mathcal{O}}

\def\be{\begin{equation}}
\def\ee{\end{equation}}
\def\beq{\be\begin{array}{c}}
\def\eeq{\end{array}\ee}
\def\bes{\be\begin{split}}
\def\ees{\end{split} \ee}
\def\bs{\begin{split}}
\def\es{\end{split} }
\def\nn{\nonumber}

\def\g{{ \gamma}}

\def\ve{{\varepsilon}}

\def\G{\Gamma}

   \makeatletter
  \let\over=\@@over \let\overwithdelims=\@@overwithdelims
  \let\atop=\@@atop \let\atopwithdelims=\@@atopwithdelims
  \let\above=\@@above \let\abovewithdelims=\@@abovewithdelims
\renewcommand\section{\@startsection {section}{1}{\z@}%
                                   {-3.5ex \@plus -1ex \@minus -.2ex}
                                   {2.3ex \@plus.2ex}%
                                   {\normalfont\large\bfseries}}

\renewcommand\subsection{\@startsection{subsection}{2}{\z@}%
                                     {-3.25ex\@plus -1ex \@minus -.2ex}%
                                     {1.5ex \@plus .2ex}%
                                     {\normalfont\bfseries}}

\begin{document}
\begin{titlepage}
\unitlength = 1mm

\vskip 1cm
\begin{center}
{ \LARGE {\textsc{Chern-Gauss-Bonnet theorem via BV localization}}}

\vspace{0.8cm}

\vspace{1cm}
 Vyacheslav Lysov\\
 \vspace{0.3cm}
 {\it London Institute for Mathematical Sciences, \\
 Royal Institution,  21 Albemarle St, London W1S 4BS, UK}\\
  \vspace{0.3cm}

\vspace{0.8cm}

\begin{abstract}
We present a new proof for the Chern-Gauss-Bonnet theorem. We represent the Euler class integral as the partition function for zero-dimensional field theory with on-shell supersymmetry. 
We rewrite the supersymmetric partition function as a BV integral and deform the Lagrangian submanifold. The new Lagrangian submanifold localizes the BV integral to the critical points of the Morse function.
\end{abstract}

\end{center}

\end{titlepage}

\pagestyle{empty}
\pagestyle{plain}

\pagenumbering{arabic}

\tableofcontents

\section{Introduction}

Chern-Gauss-Bonnet theorem \cite{chern1945curvatura} is an equality between the curvature tensor integral and the Euler characteristics of the compact manifold. Among many different proofs for this theorem, there are proofs that rely on supersymmetry. The distinct feature of such proofs \cite{witten1982supersymmetry,alvarez1983supersymmetry} is that they serve as a simple example for various connections between the differential geometry and quantum field theory (QFT).

The key observation in \cite{witten1982supersymmetry,alvarez1983supersymmetry} was that the Hodge-Dirac operator emerges as an action of supersymmeties on a Hilbert space of the $\mathcal{N}=2$ supersymmetric quantum mechanics (1-dimensional QFT). The index of the Hodge-Dirac operator, which equals the Euler characteristic, can be expressed via the twisted partition function in supersymmetric quantum mechanics. The partition function is an integral over the loop space, known as the path integral in quantum mechanics. The key feature of the $\mathcal{N}=2$ supersymmetric quantum mechanics is the four-fermion term with the Riemann tensor coupling.  
The supersymmetry allows for a deformation of quantum mechanics by a superpotential, given by a Morse function, without changing the partition function. Hence, in the limit of a large deformation, the partition function simplifies to the weighted sum over the critical point of the Morse function. 

The success of the quantum mechanical methods for the proof of Chern-Gauss-Bonnet theorem inspired similar constructions for the zero-dimensional supersymmetric QFTs, see \cite{Cordes:1994fc, Szabo:1996md,berwick2015chern,pestun2017introduction}. The crucial difference is that the path integral for zero-dimensional supersymmetric QFT is just an ordinary finite-dimensional integral over the supermanifold. 

We describe a new proof of the Chern-Gauss-Bonnet theorem based on zero-dimensional supersymmetric QFT. The main difference from existing proofs is that we use an on-shell supersymmetric theory. The supersymmetry algebra for such theory closes modulo equations of motion. The advantage of the on-shell theory is that the corresponding QFT contains only variables that are relevant to a problem. The even variables are coordinates on a manifold, while the odd variables are responsible for the integral representation of the Pfaffian and the differential form integration. The disadvantage is that we need to generalize the supersymmetric localization to the case of on-shell supersymmetry. 

In our work \cite{Losev:2023gsq} we used the Batalin-Vilkovisky (BV) formalism \cite{Batalin:1981jr, Batalin:1983ggl} to generalize the supersymmetric localization to a quadratic class on-shell supersymmetric models. We show that the zero-dimensional supersymmetric QFT for the Euler class integral belongs to the quadratic class and we can apply the BV localization. The key feature of the BV localization is the ability to deform the Lagrangian submanifold without changing the value of the integral. We use the Lagrangian submanifold, constructed from the Morse function, to localize the integral to the weighted sum over critical points. The sum, according to the Morse theorem, equals the Euler characteristic. 

Our proof provides a non-trivial consistency test for the novel BV localization method proposed in \cite{Losev:2023gsq}. Furthermore, there is a potential to modify the proof so that the four fermion partition function will emerge as the BV localization to a non-linear Lagrangian submanifold.

\section{Chern-Gauss-Bonnet theorem}

The Euler characteristic $\chi(X)$ for a smooth manifold $X$ is defined in terms of its de Rham cohomology via
\be
\chi(X) = \sum_{k=0}^{\dim X} (-1)^{k} \dim H_{dR}^k(X).  
\ee
Let us recall the familiar theorem relating two-dimensional topology and metric curvature.
\begin{Theorem}{\bf (Gauss-Bonnet)} Let $\Sigma$ be two-dimensional compact, closed manifold with Riemann metric $g$, then the integral of the Ricci scalar $R$ is related to the Euler characteristics $\chi(\Sigma)$ of the surface $\Sigma$ via
\be\label{GB_formula}
\frac{1}{2\pi}\int_\Sigma\; R \sqrt{g} = \chi(\Sigma).
\ee
\end{Theorem}
Moreover, there is a generalization of the Gauss-Bonnet theorem for even-dimensional smooth manifolds.
\begin{Theorem} {\bf (Chern-Gauss-Bonnet)} Let $X$ be compact closed $2n$-dimensional Riemann manifold 
\be
\chi(X) = \int_X e(\mathcal{R}), 
\ee
where $\mathcal{R}$ is the curvature form of (any) metric connection on the tangent bundle.
\end{Theorem}

The first proof of the theorem was given by Chern \cite{chern1945curvatura}. There are multiple proofs   \cite{Cordes:1994fc,Szabo:1996md,berwick2015chern,pestun2017introduction} relying on supersymmetry in physics inspired by  the works of Witten \cite{witten1982supersymmetry} and Alvarez-Gaume \cite{alvarez1983supersymmetry}.  The supersymmetric proofs express the Euler characteristics using the critical points for the superpotential, the Morse function on $X$. Our proof is no exception, and we will also use Morse theory results.

\section{Morse theory}

For a smooth manifold $X$ we consider a smooth function $h: X\to \mathbb{R}$, such that the critical points $\hbox{Crit}(h) = \{ p\in X\;|\; dh(p)=0\}$ of $h$ is a finite set of points. 
Critical point $p\in \hbox{Crit}(h)$ is a {\it non-degenerate critical point} if the Hessian matrix $\p_i\p_j h(p)$ has maximal rank. The non-degenerate critical points are isolated. A Morse function is a smooth function with non-degenerate critical points only. By Morse lemma, we can arrange the coordinates in the vicinity of a non-degenerate critical point, such that the Hessian matrix takes a standard form $\hbox{diag} (1,...,1,-1,...,-1)$. The number of negative entries is the {\it Morse index} $\g(p)$ of the critical point $p$. We denote  the set of critical points with Morse index $\g(p)=k$ by $Crit_k(h)$.

\begin{Theorem}\label{thm_Morse}{\bf (Morse)} For a smooth function $h: X \to \mathbb{R}$ with non-degenerate critical points the following holds
\begin{itemize}
\item Morse inequality for de Rham cohomology 
\be
\dim H_{dR}^k(X) \leq  |Crit_k(h)|.
\ee
\item Morse equality for Euler characteristic
\be
\begin{split}
\chi(X) = \sum_{k=0}^{\dim X} (-1)^k |Crit_k(h)|.
\end{split}
\ee
\end{itemize}

\end{Theorem}

\section{Superspace integral representation for Euler class}

The first step of the Chern-Gauss-Bonnet theorem's  proof  is to represent the Euler class  as a supermanifold integral.  Certain signs in the integral representation  will depend on the conventions for  the integration over the odd variables.   We define a  Berezin integral over  multiple odd variables  $\theta_1,\ldots, \theta_m$ via 
\be\label{eq_Ber_int_definition}
\int D\theta_m\cdots D\theta_1 \; f(\theta_1,\ldots,\theta_m)  =  \frac{\p}{\p \theta^m} \cdots   \frac{\p}{\p \theta^2} \frac{\p}{\p \theta^1}  f. 
\ee
The  integration measure in the integral (\ref{eq_Ber_int_definition}) is commonly referred to as the coordinate Berezinian, which we denote as 
\be
D^m \theta =  D\theta_m\cdots D\theta_1.
\ee

\begin{Lemma}\label{lemma_superspace_integral_rep} For an arbitrary Riemann metric  $g$ on  a smooth $2n$-dimensional compact manifold $X$ the  Euler class has a superspace integral representation 
\be\label{eq_Eucler_class_superspace_integral}
\int_X  e (\mathcal{R})  =  \frac{(-1)^n}{(2\pi)^{n}}\int_{ T[1]X\oplus T[1]X}  d^{2n} x\; D^{2n}\psi D^{2n} \bar{\psi} \;  (\det g)^{-\frac12}\;e^{- \frac{1}{4}R_{abcd} \psi^a \psi^b \bar{\psi}^c \bar{\psi}^d},
\ee
where $\mathcal{R}_{ab}$  is the   curvature 2-form, related to the Riemann curvature tensor $R_{abcd}$ via
\be
\mathcal{R}_{ab} =  \frac12 \sum_{cd}R_{abcd} \;dx^c \wedge dx^d.
\ee
\end{Lemma}
\begin{proof} The  definition of the Euler class is 
\be\label{eq_def_euler_class}
e (\mathcal{R}) = \frac{1}{(2\pi)^n} \hbox{Pf} (\mathcal{R}).
\ee
The  Pfaffian  of the matrix has a Berezin integral (\ref{eq_Ber_int_definition})  representation
\be\label{eq_pfaff_int_repr}
\hbox{Pf} (A) = \frac{(-1)^n}{n! 2^n} \epsilon^{a_1b_1a_2b_2...a_nb_n}  A_{a_1b_1}....A_{a_nb_n} = \int_{\mathbb{R}^{0|2n}} D^{2n}\bar{\psi}\;\; e^{-\frac12 A_{ab} \bar{\psi}^a\bar{\psi}^b}.
\ee
Using the Pfaffian representation (\ref{eq_pfaff_int_repr}) the Euler class (\ref{eq_def_euler_class}) is
\be\label{eq_euler_class_integral}
e (\mathcal{R})  = \frac{(-1)^n}{(2\pi)^n}  \int_{\mathbb{R}^{0|2n}} D^{2n}\bar{\psi}\;\; e^{-\frac12 \mathcal{R}_{ab} \bar{\psi}^a\bar{\psi}^b}.
\ee
We can include the factor of $\sqrt{\det g}$ and change the integration  over  $\mathbb{R}^{0|2n}$ in (\ref{eq_euler_class_integral})  to the integration over  fibers of $T[1]X$, i.e.
\be\label{eq_euler_class_integral_fibers}
e (\mathcal{R})  = \frac{(-1)^n}{(2\pi)^n}  \int_{T[1]X / X} D^{2n}\bar{\psi} \;\; (\det g)^{-\frac12}\;\; e^{-\frac12 \mathcal{R}_{ab} \bar{\psi}^a\bar{\psi}^b}.
\ee
The space of differential forms $\Omega^\ast(X)$ on $X$  can be identified with the space of functions on a supermanifold $T[1]X$, while the integration   over $X$ of  differential form $\omega(x, dx)$  is identical to the  integration of the corresponding function $\omega (x,\psi)$ over the supermanifold  $T[1]X$ with the standard Berezinian, i.e. 
\be\label{eq_forms_superspace_integral}
\int_X \omega (x, dx) = \int_{ T[1]X} d^{2n}x D^{2n}\psi\;\; \omega(x,\psi).
\ee
Using an integral representation (\ref{eq_euler_class_integral_fibers})   and a superspace representation (\ref{eq_forms_superspace_integral}) for the differential forms integration  we arrive at equality (\ref{eq_Eucler_class_superspace_integral}) of the lemma, hence the proof is complete.
\end{proof}

\begin{Remark}
There are different signs, used in the literature for the Chern-Gauss-Bonnet theorem, due to different conventions for the Berezin integration (\ref{eq_Ber_int_definition}) and Pfaffian definition (\ref{eq_pfaff_int_repr}). 
The consistency of signs can be checked  using   $X = \mathbb{S}^{2n}$, the $2n$-dimensional round sphere of radius $\ell$. 
The  Euler characteristic is $\chi (\mathbb{S}^{2n}) = 2$.  

The $\mathbb{S}^{2n}$  is a maximally symmetric space, so the  Riemann curvature tensor  simplifies to 
\be
R_{abcd} = \frac{R}{2n(2n-1)} (g_{ac}g_{bd} - g_{bc} g_{ad}) = \frac{1}{\ell^2}(g_{ac}g_{bd} - g_{bc} g_{ad}),
\ee
where $g_{ab}$ is the metric components and $R = \frac{2n(2n-1)}{\ell^2} $ is the Ricci curvature.

The exponent in  (\ref{eq_Eucler_class_superspace_integral}) simplifies to 
\be
 -\frac{1}{4}R_{abcd} \psi^a \psi^b \bar{\psi}^c \bar{\psi}^d  = \frac{1}{2\ell^2} (g_{ab}\psi^a \bar{\psi}^b)^2.
\ee

The integral (\ref{eq_Eucler_class_superspace_integral}) evaluates into 
\be
\begin{split}
\chi(\mathbb{S}^{2n})  &=  \frac{(-1)^n}{(2\pi)^{n}}\int_{ T[1]\mathbb{S}^{2n}\oplus T[1]\mathbb{S}^{2n}}  d^{2n} x\; D^{2n}\psi D^{2n} \bar{\psi} \;  (\det g)^{-\frac12}\;e^{- \frac{1}{4}R_{abcd} \psi^a \psi^b \bar{\psi}^c \bar{\psi}^d}\\
& =  \frac{(-1)^n}{(2\pi)^{n}}  \int_{\mathbb{S}^{2n}} d^{2n}x   (\det g)^{-\frac12}\; \frac{(-1)^n (2n)!}{\ell^{2n}2^n n!}  \cdot  \det g  =\frac{(2n)!}{\ell^{2n}(4\pi)^n n!}  \cdot \hbox{Vol} (\mathbb{S}^{2n})\\
& = \frac{(2n)!}{\ell^{2n}(4\pi)^n n!}  \cdot  \frac{(4\ell^2 \pi)^n (n-1)!}{(2n-1)!} = \frac{2n}{n} =2.
  \end{split}
\ee
\end{Remark}

\section{Supersymmetric partition function}

The curvature term in (\ref{eq_Eucler_class_superspace_integral}) is the feature of  $\mathcal{N}=2$ supersymmetric quantum mechanical sigma-models \cite{witten1982supersymmetry,alvarez1983supersymmetry} on curved manifolds.
Moreover, the curvature term is the only term that remains if we reduce the dimension of the source space from one to zero. Hence, we expect it to be invariant under the zero-dimensional  version of the supersymmetry transformations. 
Furthermore, we can interpret the integral (\ref{eq_Eucler_class_superspace_integral}) for the Euler class  as the supersymmetric partition function. 

\begin{Lemma}\label{lemma_part_funct_rep} The superspace integral (\ref{eq_Eucler_class_superspace_integral})  is the  partition function for an $\mathcal{N}=2$ on-shell supersymmetric theory, i.e.
\be\label{eq_supersymmetric_part_function}
\int_X e(\mathcal{R})  =\frac{(-1)^n}{(2\pi)^n}   \int_{M} \mu_M\;\; e^{-S}   
\ee
 on a supermanifold $M = T[1]X \oplus T[1]X$  with coordinates $x^i, \psi^i, \bar{\psi}^i$, classical action
\be\label{eq_4fermion_action}
S[x,\psi,\bar{\psi}] = \frac14 R_{abcd}  \psi^a  \psi^b \bar{\psi}^c   \bar{\psi}^d
\ee
and the  integration measure $\mu_M = (\det g)^{-\frac12} \;\;d^{2n}x D^{2n} \psi D^{2n}\bar{\psi}$.
\end{Lemma}
\begin{proof}  The  $\mathcal{N}=2$ supersymmetric partition function is a superspace integral of an exponential function of the classical action, such that the action $S$ and the integration measure $\mu_M$ 
are invariant under the action of the two odd  vector fields ${Q}$ and $\bar{Q}$. Moreover,  the vector fields are  nilpotent and anticommute on the critical set  of $S$.

We use the  supersymmetric quantum mechanics results \cite{witten1982supersymmetry,alvarez1983supersymmetry} as a motivation for the supersymmetry vector fields
\be\label{eq_supersymmetry_vector_fields}
Q = \bar{\psi}^j \frac{\p}{\p x^j} +\psi^k\bar{\psi}^l  \G_{kl}^j \frac{\p }{\p \psi^j},\;\; \bar{Q} = \psi^j \frac{\p}{\p x^j} - \psi^k\bar{\psi}^l  \G_{kl}^j \frac{\p }{\p \bar{\psi}^j}.  
\ee
The action (\ref{eq_4fermion_action})  is invariant  under the supersymmetry  vector fields, i.e. 
\be
Q(S) = \bar{Q}(S)  =0.
\ee
Indeed,  the vector  field's action evaluates into 
\be
\begin{split}
Q (S) &= \left( \bar{\psi}^j \frac{\p}{\p x^j} +\psi^k\bar{\psi}^l  \G_{kl}^j \frac{\p }{\p \psi^j} \right)R_{abcd}  \psi^a  \psi^b \bar{\psi}^c   \bar{\psi}^d \\
& = \p_j R_{cdab}  \bar{\psi}^j  \bar{\psi}^c   \bar{\psi}^d  \psi^a  \psi^b +2   \G_{aj}^e  R_{ cd be} \bar{\psi}^j   \bar{\psi}^c   \bar{\psi}^d   \psi^a  \psi^b= \nabla_{[j} R_{ab]cd}  \bar{\psi}^j   \bar{\psi}^a   \bar{\psi}^b   \psi^c  \psi^d= 0.
\end{split}
\ee
We used the symmetry properties of the Riemann tensor  $R_{abcd} = - R_{bacd} =  - R_{cdab}$ and the Bianchi identity  $\nabla_{[j} R_{ab]cd} = 0$ in the last equality. The vanishing  of $\bar{Q}(S)$ is proven in a similar way.

 The invariance of the integration measure $\mu_M$ is equivalent to the divergence-free condition for the   supersymmetry vector fields, i.e.
 \be
 \hbox{div}_{\mu_M} Q =  \hbox{div}_{\mu_M} \bar{Q}  = 0.
 \ee
 The divergence of the vector field $v$  with respect to the measure $\mu = \rho \cdot \mu_c$  is  given by 
\be
\hbox{div}_{\mu_M} v = \hbox{div}_{\mu_c} v + v (\ln \rho), 
\ee
For the vector field 
\be
v = v^a \frac{\p}{\p x^a} + u^a  \frac{\p}{\p \psi^a} +\bar{u}^a  \frac{\p}{\p \bar{\psi}^a}
\ee
 the divergence in the standard measure $\mu_c =  d^{2n}x D^{2n}\psi D^{2n}\bar{\psi}$ is 
\be
\hbox{div}_{\mu_c} v = \p_{x^a}  v^a- (-1)^{|v|} \p_{\psi^a} u^a - (-1)^{|v|} \p_{\bar{\psi}^a} \bar{u}^a. 
\ee
For the vector fields (\ref{eq_supersymmetry_vector_fields}) the divergence evaluates into
\be
\begin{split}
\hbox{div}_{\mu_X} Q &= \hbox{div}_{\mu_c} Q + Q (\ln (\det g)^{-\frac12})  = \frac12 \bar{\psi}^a \p_a  \ln \det  g -\frac12 \bar{\psi}^a \p_a  \ln \det g = 0,\\
\hbox{div}_{\mu_X} \bar{Q} &= \hbox{div}_{\mu_c} \bar{Q} + \bar{Q} (\ln (\det g)^{-\frac12})  = \frac12 \psi^a \p_a  \ln  \det g - \frac12 \psi^a \p_a  \ln \det g  =0.
\end{split}
\ee
The algebra of supersymmetry  vector fields (\ref{eq_supersymmetry_vector_fields}) is
\be\label{eq_on_shell_algebra}
\begin{split}
[Q, Q\} &= -\psi^b \bar{\psi}^c \bar{\psi}^d R_{~bcd}^{a}  \p_{\psi^a} =-2 g^{ab}\; \p_{\psi^b} S  \;\p_{\psi^a} ,\\
 [\bar{Q}, \bar{Q}\} &=- \psi^c \psi^d\bar{\psi}^bR_{~bcd}^{a} \p_{\bar{\psi}^a}  =-2 g^{ab} \;\p_{\bar{\psi}^b} S\;  \p_{\bar{\psi}^a} ,\\
 [ \bar{Q},Q\} &=  \psi^c \psi^d\bar{\psi}^bR_{~bcd}^{a}  \p_{\psi^a}  +  \bar{\psi}^c \bar{\psi}^d  \psi^bR_{~bcd}^{a}   \p_{\bar{\psi}^a}  =-g^{ab}\; \p_{\psi^b} S  \;\p_{\bar{\psi}^a} -g^{ab}\; \p_{\bar{\psi}^b} S  \;\p_{\psi^a}.
\end{split}
\ee
The right-hand side of the graded commutators in (\ref{eq_on_shell_algebra}) vanish on the critical set $M_{crit} = \{ \phi \in M\;|\; dS(\phi) = 0\}$ hence the vector fields become nilpotent and anticommuting.   
We checked all required properties of the supersymmetric partition function, so the proof is complete.

\end{proof}
In  physics  terminology the supersymmetry algebra (\ref{eq_on_shell_algebra}) is the  {\it on-shell $d=0$ $\mathcal{N}=2$ supersymmetry algebra}.  The number $\mathcal{N}$ refers to the number of independent supersymmetries, i.e. the dimension of the corresponding
odd part of the superalgebra.  The term {\it on-shell} indicates that the nontrivial  commutators of supercharges  are proportional to the equations of motion for 
the action  (\ref{eq_4fermion_action}), i.e.
\be
\frac{\p S}{ \p\psi^j} = \frac12 R_{ j bcd}    \psi^b \bar{\psi}^c   \bar{\psi}^d,\;\;\; \frac{\p S}{ \p \bar{\psi}^j} =  \frac12 R_{ab j d}  \psi^a  \psi^b    \bar{\psi}^d.
\ee
Therefore, as long as equations of motions are satisfied, the on-shell supersymmetry algebra is the same as the off-shell one.
\be
[Q, Q\} = [\bar{Q}, Q\} = [\bar{Q}, \bar{Q}\} = 0.
\ee

\section{BV formalism}

For a supermanifold $M$ there is an odd symplectic manifold   $\mathcal{M} = T^\ast[1]M$ with the fiber coordinates known as the antifeilds. In particular, for $M  = T[1]X\oplus T[1]X$ with coordinates $x, \psi, \bar{\psi}$ the corresponding antifields are $x^\ast, \psi^\ast, \bar{\psi}^\ast$.  The canonical  symplectic form on $\mathcal{M} $ is 
\be\label{sympl_form}
\omega = dx^i\wedge dx^\ast_i  - d\psi^i\wedge d\psi^\ast_i - d\bar{\psi}^i\wedge d\psi^\ast_i.
\ee
The BV bracket  for two functions $f,g$ with parities $|f|, |g|$  is
\be
\begin{split}
\{ f, g\}  &=  \p_{x^i} f\; \p_{x^\ast_i} g +   (-1)^{|f|} \p_{x^\ast_i} f\; \p_{x^i} g -  \p_{\psi^\ast_i} f\; \p_{\psi^i} g -(-1)^{|f|} \p_{\psi^i} f\; \p_{\psi^\ast_i} g\\
&\qquad  -  \p_{\bar{\psi}^\ast_i} f\; \p_{\bar{\psi}^i} g -(-1)^{|f|} \p_{\bar{\psi}^i} f\; \p_{\bar{\psi}^\ast_i} g.
\end{split}
\ee
A   Berezinian 
\be
\mu = \rho(x,\psi, \bar{\psi})\; d^{2n} x D^{2n}\psi D^{2n}\bar{\psi} D^{2n}x^\ast d^{2n}\psi^\ast d^{2n} \bar{\psi} 
\ee
 on $\mathcal{M}$ is a compatible Berezinian on $(\mathcal{M}, \omega)$, while the corresponding    BV Laplacian is
 \be
 \Delta_\mu =    \frac{\p}{\p x^i} \frac{\p}{\p x^\ast_i} -  \frac{\p}{\p \psi^i} \frac{\p}{\p \psi^\ast_i}-  \frac{\p}{\p \bar{\psi}^i} \frac{\p}{\p \bar{\psi}^\ast_i} + \frac12\{\ln \rho, \cdot\}.
 \ee
\begin{Definition} For  an odd symplectic space with compatible Berezinian   $(\mathcal{M}, \omega, \mu)$ an even function  $\mathcal{S}\in C^\infty(\mathcal{M})$ is a solution to the     quantum master equation (QME) if it satisfies
\be
\label{quantum_mast_eq}
\frac12  \{\mathcal{S}, \mathcal{S}\} =   \Delta_\mu \mathcal{S}=0. 
\ee
\end{Definition}
In physics literature, a solution to QME is often called a BV action. Given a QME solution $\mathcal{S}$ there is a $\Delta_\mu$-closed exponential function $f = e^{-\mathcal{S}}$.

\begin{Remark} Our definition (\ref{quantum_mast_eq}) for the quantum master equation   differs from the commonly-used  definition  
\be\label{eq_liter_quantum_mast_eq}
\frac12 \{ \mathcal{S}, \mathcal{S}\}  = \hbar \Delta_\mu \mathcal{S} 
\ee
 in the BV literature \cite{Batalin:1981jr,Batalin:1983ggl,Schwarz:1992nx,mnev2019quantum}.  Our solution (\ref{quantum_mast_eq})  to the QME is the $\hbar$-independent solution  of (\ref{eq_liter_quantum_mast_eq}). An additional restriction on a solution ensures that the exponent 
 $f = e^{- \frac{1}{\hbar} \mathcal{S}}$ is not a formal series expansion in $\hbar$, but rather a well-defined function.
\end{Remark}

\begin{Definition} A {\it Lagrangian submanifold} $\mathcal{L}$ of an odd symplectic manifold $(\mathcal{M}, \omega)$ is 
\begin{itemize}
\item isotropic, i.e. $\omega|_{\mathcal{L}} = 0$;
\item not a proper submanifold of another isotropic submanifold of $\mathcal{M}$.
\end{itemize} 
\end{Definition}
For an odd function $V$ on $M$,  there is an  associate  Lagrangian submanifold
 \be
 \cl_V  = \hbox{graph} (dV) \subset T^\ast[1]M.
 \ee
 For a supermanifold $\mathcal{N}$   there is a  canonical map sending Berezinians $\mu$ on $T^\ast[1] \mathcal{N}$ to Berezinians $\sqrt{\mu|_\mathcal{N}}$ on $\mathcal{N}$. Locally, for  local coordinates $X^a$
 on $\mathcal{N}$, coordinates $(X^a, \Xi_a)$  on $T^\ast[1] \mathcal{N}$ and a Berezinian $\mu = \rho(X, \Xi) d^nX d^n\Xi$ the corresponding Berezinian in $\mathcal{N}$
 is $\sqrt{\mu|_\mathcal{N}} = \sqrt{\rho(X, 0)} d^nX$.
 
 \begin{Definition}\label{def_bv_integral}
 For odd symplectic manifold with compatible Berezinian  $(\mathcal{M}, \omega, \mu)$ a {\it BV-integral} of a function $f \in C^\infty(\mathcal{M})$, such that $\Delta_\mu f = 0$, over Lagrangian submanifold $\mathcal{L}$    is an integral 
 \be\label{eq_def_bv_int}
 \int_{\mathcal{L}} f\; \sqrt{\mu|_{\mathcal{L}}}.
 \ee
 \end{Definition}

\begin{Theorem}\label{thm_bv_invariance} {\bf (Batalin-Vilkovisky-Schwarz \cite{Schwarz:1992nx})}  For an odd-symplectic manifold with compatible  Berezinian  $(\mathcal{M}, \omega, \mu)$ and  compact body.
\begin{enumerate}
\item For any $g \in C^\infty(\mathcal{M})$  and Lagrangian submanifold  $\mathcal{L} \subset \mathcal{M}$ the BV integral 
 \be
 \int_{\mathcal{L}} \Delta_\mu g\; \sqrt{\mu|_{\mathcal{L}}} = 0.
 \ee
 \item For a pair of Lagrangian submanifolds $\mathcal{L}$ and $\mathcal{L}'$ with homologous bodies in the body of $\mathcal{M}$ and a function, such that $\Delta_\mu f = 0$ the following holds
  \be
 \int_{\mathcal{L}}  f\; \sqrt{\mu|_{\mathcal{L}}} =   \int_{\mathcal{L}'}  f\; \sqrt{\mu|_{\mathcal{L}'}} .
\ee
\end{enumerate}
\end{Theorem}

\section{BV description of the curvature theory}

According to results of    \cite{Losev:2023gsq}    a supersymmetric partition function  (\ref{eq_supersymmetric_part_function}) defines an  approximate solution to the QME 
\be\label{eq_approx_solution_qme}
\mathcal{S}(\phi, \phi^\ast) = S(\phi) +  \mathcal{Q} + \Pi^{(2)}.
\ee
The first term in (\ref{eq_approx_solution_qme}) is  the classical action   (\ref{eq_4fermion_action}) for the supersymmetric system.  The linear term in antifields,   $\mathcal{Q}$,  is the generator of supersymmetries  (\ref{eq_supersymmetry_vector_fields}) 
\be
\mathcal{Q} = (  \ve \bar{\psi}^j +\bar{\ve} \psi^j )  x_j^\ast +  \psi^k\bar{\psi}^l  \G_{kl}^j  (  \ve \psi^\ast_j - \bar{\ve} \bar{\psi}^\ast_j ).
\ee
The quadratic term in antifields  $\Pi^{(2)}$ is a  BV bivector, given by the commutation relations (\ref{eq_on_shell_algebra})
\be\label{eq_bivector}
\Pi^{(2)} =   \frac12 g^{kj} (  \ve \psi^\ast_j -\bar{\ve} \bar{\psi}^\ast_j )  (  \ve \psi^\ast_k -\bar{\ve} \bar{\psi}^\ast_k ). 
\ee

Note that the bivector is not a constant bivector as in  simple examples of \cite{Losev:2023gsq}. However it satisfies the necessary properties, so the QME solution (\ref{eq_approx_solution_qme}) is a quadratic type theory according to the refined classification of \cite{Losev:2023gsq}. 

\begin{Proposition}\label{prop_quantum_master_solution} The BV action 
\be\label{eq_BV_action_curvature_model}
\begin{split}
\mathcal{S}  &= \frac14 R_{abcd} \psi^a \psi^b \bar{\psi}^c \bar{\psi}^d + (  \ve \bar{\psi}^j +\bar{\ve} \psi^j )x_j^\ast  +    \psi^k\bar{\psi}^l  \G_{kl}^j (  \ve \psi^\ast_j - \bar{\ve} \bar{\psi}^\ast_j )
\\&\qquad\qquad\qquad + \frac12 g^{kj} (  \ve \psi^\ast_j -\bar{\ve} \bar{\psi}^\ast_j )  (  \ve \psi^\ast_k -\bar{\ve} \bar{\psi}^\ast_k ) 
\end{split}
\ee
solves the quantum master equation on an odd  symplectic space  $\mathcal{M} = T^\ast[1]M$ with canonical symplectic form (\ref{sympl_form}) and  a compatible Berezinian
\be
\mu  =  \frac{1}{\det g}  d^{2n} x  D^{2n}\psi  D^{2n} \bar{\psi}  D^{2n} x^\ast  d^{2n} \psi^\ast d^{2n}\bar{\psi}^\ast.
\ee
\end{Proposition}
\begin{proof} The statement of proposition is   that the  approximate solution (\ref{eq_approx_solution_qme}) does not aquire higher order terms in antifields. Hence, we need to check the QME for orders 
$\cO(\phi^\ast)^2$, $\cO(\phi^\ast)^3$  and the triviality of BV Laplacian for the bivector (\ref{eq_bivector}).

 The $\cO(\phi^\ast)^2$ order QME 
\be
\begin{split}
\{ \mathcal{Q}, \Pi^{(2)}\} &= \{ (  \ve \bar{\psi}^j +\bar{\ve} \psi^j )x_j^\ast , \Pi^{(2)}\} + \{\psi^k\bar{\psi}^l  \G_{kl}^j ( \ve \psi^\ast_j -\bar{\ve} \bar{\psi}^\ast_j ), \Pi^{(2)} \}   \\
& = -\frac12 \left( \p_j g^{kl}  +2g^{kc}  \G_{c j}^l \right) (\ve \bar{\psi}^j   + \bar{\ve}\psi^j)  ( \ve \psi^\ast_l -\bar{\ve} \bar{\psi}^\ast_l)(  \ve \psi^\ast_k -\bar{\ve} \bar{\psi}^\ast_k ) = 0. 
\end{split}
\ee
The last equality follows from the symmetrization over $kl$ for  the identity 
\be
\begin{split}
 \p_j g^{kl} + 2g^{kc}  \G_{c j}^l =   g^{ka} g^{lb} (\p_a g_{bj}- \p_b g_{aj}) 
\end{split}
\ee
 The $\cO(\phi^\ast)^3$ order QME  
\be
\{ \Pi^{(2)}, \Pi^{(2)}\} = 0
\ee
is rather trivial, since the bivector has only components along the fermionic antifields $\psi^\ast, \bar{\psi}^\ast$, while the coefficients, for these components depend on $x$ only.

The BV Laplacian for bivector 
\be
\Delta_\mu \Pi^{(2)} = \Delta_{\mu_c} \Pi^{(2)} + \frac12\{\ln \det g(x), \Pi^{(2)}\} =  0
\ee
is almost trivial since the determinant of metric $\det g$ has only $x$-dependence while the $ \Pi^{(2)}$ is independent on $x^\ast$.
\end{proof}

\section{BV localization}

The supersymmetric partition function (\ref{eq_supersymmetric_part_function}) has a  BV integral (\ref{eq_def_bv_int}) representation.  In particular, it equals to the  QME solution (\ref{eq_BV_action_curvature_model})  integrated  over trivial Lagrangian submanifold $\cl_0$, the zero section of the $T^\ast[1] M$ bundle
\be\label{eq_susy_part_function_trivial_bv_integral}
\frac{1}{(2\pi)^n}   \int_{M} \mu_M\;\; e^{-S}  =   \frac{1}{(2\pi)^{n}} \int_{\cl_{0}} \sqrt{\mu}\;\; e^{-  \mathcal{S}}. 
\ee

For a Morse function $W:X \to \mathbb{R}$, we  define a Lagrangian submanifold $\cl_V$, associated to with odd function $V = t \psi^i \p_i W $. In our coordinates
\be\label{eq_lagr_submanifold_large_antifields}
\mathcal{L}_{V} = \left\{ x^\ast_i = t \psi^j \p_i\p_j W(x),\;\; \psi^\ast_i  = t \p_i W(x),\;\; \bar{\psi}^\ast_i = 0 \right\} \subset T^\ast[1] M.
\ee

\begin{Lemma}\label{lemma_bv_localiz} The $t \to \infty$ limit of the  BV integral   for the QME solution (\ref{eq_BV_action_curvature_model}) with $\ve =1, \; \bar{\ve}=0$ over the  Lagrangian submanifold (\ref{eq_lagr_submanifold_large_antifields}),  is expressed via Morse indices of $W$, i.e.
\be\label{eq_Bv_integral_limit_weighted_sum}
I_V =  \frac{(-1)^n}{(2\pi)^{n}} \int_{\cl_{V}} \sqrt{\mu}\;\; e^{-  \mathcal{S}} = \sum_{p \in Crit(W)} (-1)^{\g(p)}.
\ee
\end{Lemma}
\begin{proof}
The restriction of the QME solution (\ref{eq_BV_action_curvature_model}) to the Lagrangian  submanifold (\ref{eq_lagr_submanifold_large_antifields}) evaluates into
\be\label{eq_large_t_lagrangian_bv_action}
\mathcal{S}\Big|_{\mathcal{L}_{V}} = \frac14  R_{abcd} \psi^a \psi^b \bar{\psi}^c \bar{\psi}^d + t \p_j \p_k W  \psi^k \bar{\psi}^j +t\psi^k\bar{\psi}^l  \G_{kl}^j   \p_j W   + \frac12  t^2 g^{kj}   \p_j W \p_k W.
\ee
The BV integral in the $t \to \infty$ simplifies since we can drop the curvature term in (\ref{eq_large_t_lagrangian_bv_action}) comparing to the $t$-enhanced  quadratic fermionic terms  
\be\nn
I_{V} = \frac{(-1)^n}{(2\pi)^n} \int_{M}  d^{2n} x D^{2n} \psi D^{2n} \bar{\psi}\;\; ( \det g)^{-\frac12} \exp \left(- t \nabla_j \p_k W  \psi^k \bar{\psi}^j  -\frac12 t^2 g^{kj}   \p_j W \p_k W \right).
\ee
The integral over fermions is Gaussian, so we can perform it exactly 
\be
\int D^{2n} \psi \; D^{2n} \bar{\psi}\;\;  \exp \left(- t \nabla_j \p_k W  \psi^k \bar{\psi}^j  \right) =  (-1)^nt^{2n} \det (\nabla_j \p_k W ).
\ee
An expression $g^{kj}   \p_j W \p_k W$ is non-negative  on $X$ and takes zero values at critical points of $W$.  Hence the integral  over $X$ for large $t$ localizes to the neighborhood of critical points of $W$, i.e.
\be\label{eq_crit_point_sum_integral}
I_{V} = \frac{(-1)^{2n}}{(2\pi)^n} \sum_{p\in \hbox{Crit}(W)}  \int_{U_p} d^{2n} x\;   ( \det g)^{-\frac12}\;  \exp \left(- \frac12 t^2 g^{kj}   \p_j W \p_k W \right) \; t^{2n} \det (\nabla_j \p_k W ).
\ee 
The integral (\ref{eq_crit_point_sum_integral}) is invariant under the coordinate transformations on $X$. By the Morse lemma, we  can choose the local coordinates, such that near the critical point $p$ Morse function $W$ takes the form 
\be
W =W(p) - \frac12 \sum_{j=1}^{\g(p)} (y^j)^2  +\frac12  \sum_{j= \g(p)+1}^{2n} (y^j)^2,
\ee 
where $\g(p)$ is the Morse index of the critical point $p$. The determinant in Morse coordinates 
\be\label{eq_fermion_det_crit_point}
 \det (\nabla_j \p_k W )(p) = (-1)^{\g(p)} .
\ee
We perform a variable change  $z^i = ty^i$ so that the argument of the exponent near the critical point $p$ simplifies 
\be
t^2 g^{kj} (x)  \frac{\p W}{\p x^j} \frac{\p W}{\p x^k} = t^2 \tilde{g}^{kj} (y) y_j y_k  = t^2 \hat{g}^{kj} (t^{-1}z) z_j z_k = t^2 \hat{g}^{kj} (0) z_j z_k + \cO(t), 
\ee
where $\tilde{g}$ is the metric in $y$ coordinates and $\hat{g}$ is the metric in $z$-coordinates. 
The  integral (\ref{eq_crit_point_sum_integral}) near  the critical point  $p$ for large $t$ evaluates into
\be\nn
 \begin{split}
I_V &= \sum_{p\in \hbox{Crit}(W)} \frac{(-1)^{\g(p)}  t^{2n}}{(2\pi)^n} \int_{\mathbb{R}^{2n}} d^{2n} z\;  (\det \hat{g})^{-\frac12}\;  \exp \left(-\frac12 t^2 \hat{g}^{kj} (0) z_j z_k +\cO(t) \right)  \\
 &=  \sum_{p\in \hbox{Crit}(W)} \frac{ (-1)^{\g(p)} t^{2n}}{(2\pi)^n}  \left( \frac{2\pi}{t^2}\right)^n (\det \hat{g}(0))^{-\frac12}  \frac{1}{\sqrt{\det \hat{g}^{-1}}(0)} (1+\cO(t^{-1})) \\
 &= \sum_{p\in \hbox{Crit}(W)} (-1)^{\g(p)} + \cO(t^{-1}).
 \end{split}
\ee
The $t \to \infty$ limit removes the $\cO(t^{-1})$ corrections and we arrive into (\ref{eq_Bv_integral_limit_weighted_sum}), so the proof is complete.
\end{proof}

\begin{Remark} We can modify the lemma \ref{lemma_bv_localiz} if we use the vector field $v$ with isolated critical points instead of the Morse vector field $g^{-1}(dW,\cdot)$. The corresponding Lagrangian submanifold is associated with the 
odd function $V = \psi^j g_{jk} v^k(x)$. The result of the BV integration in (\ref{eq_Bv_integral_limit_weighted_sum}) over this Lagrangian submanifold equals the sum over zeros of the vector field, weighted by the corresponding indices. 
The sum is identical to the sum in the Poincare-Hopf theorem, representing the Euler characteristics of $X$.
\end{Remark}

\section{Proof of the Chern-Gauss-Bonnet theorem}

 The first step of the proof is to use the lemma \ref{lemma_superspace_integral_rep} and turn the Euler class integral over $X$ into the supermanifold integral over $M = T[1]X \oplus T[1]X$.
The next step is to use the lemma \ref{lemma_part_funct_rep} to represent the supermanifold integral as a supersymmetric partition function with on-shell supersymmetry. 
The third step is to rewrite (\ref{eq_susy_part_function_trivial_bv_integral}) the supersymmetric partition function as a BV integral over the trivial  Lagrangian submanifold for the  BV action (\ref{eq_BV_action_curvature_model}). In particular, the  proposition \ref{prop_quantum_master_solution} implies that the 
 on-shell supersymmetric theory with action (\ref{eq_4fermion_action}) is a quadratic type of BV supersymmetric theory in refined classification from \cite{Losev:2023gsq}. 
 
The key result in BV formalism is that the BV integral for the solution of the quantum master equation is invariant under the deformations of the Lagrangian submanifold, formalized by the theorem \ref{thm_bv_invariance}.
 The lemma \ref{lemma_bv_localiz} describes the BV integral over limiting Lagrangian submanifold, defined by the Morse function $W$ on $X$ and expresses the final result in terms of Morse index for the function $W$. The Morse theorem \ref{thm_Morse} implies a relation between the weighted sum in (\ref{eq_Bv_integral_limit_weighted_sum}) and the Euler characteristic of $X$ and completes the proof.

\section*{Acknowledgments}

V.L. is grateful to Yasha Neiman, Andrey Losev, and Pavel Mnev for the useful discussions on the material of this paper.  The work of V.L. is supported by the Arnold Fellowship at the London Institute for Mathematical Sciences.

\bibliography{gauss_ref}{}
\bibliographystyle{utphys}

\end{document}